\documentclass[onecolumn,11pt,a4paper,floatfix,accepted=2024-05-01,superscriptaddress]{quantumarticle}
\pdfoutput=1
\usepackage{amsfonts,amsmath, amssymb, amsthm, algorithm,algpseudocode, enumitem, thm-restate, bbm}
\usepackage{authblk}
\usepackage{graphicx}
\usepackage{array}
\usepackage{tikz-cd}
\usepackage{braket}
\usepackage{scalerel}
\PassOptionsToPackage{obeyspaces}{url}
\usepackage[backend=biber,style=alphabetic,url = false]{biblatex}
\usepackage[colorlinks=true,citecolor=blue,urlcolor=blue,linkcolor=blue,bookmarksopen=true]{hyperref}
\author{Andrew Cross}
\affiliation{IBM Quantum, IBM T.J. Watson Research Center, Yorktown Heights, NY, United States.} 
\author{Zhiyang He (Sunny)}
\email{szhe@mit.edu}
\affiliation{IBM Quantum, IBM T.J. Watson Research Center, Yorktown Heights, NY, United States.} 
\affiliation{Department of Mathematics, Massachusetts Institute of Technology.}
\author{Anand Natarajan}
\affiliation{Computer Science and Artificial Intelligence Laboratory, Massachusetts Institute of Technology.} 
\author{Mario Szegedy}
\affiliation{Rutgers, The State University of New Jersey, New Brunswick, NJ, United States.}  
\author{Guanyu Zhu}
\affiliation{IBM Quantum, IBM T.J. Watson Research Center, Yorktown Heights, NY, United States.}

\title{Quantum Locally Testable Code with Constant Soundness}
\date{}

\newtheorem{theorem}{Theorem}[section]

\newtheorem{claim}[theorem]{Claim}
\newtheorem{fact}[theorem]{Fact}
\newtheorem{definition}[theorem]{Definition}

\newtheorem{lemma}[theorem]{Lemma}
\newtheorem{corollary}[theorem]{Corollary}

\DeclareMathOperator{\Span}{Span}

\DeclareMathOperator{\css}{CSS}

\DeclareMathOperator{\im}{im}
\DeclareMathOperator{\avg}{avg}

\newcommand*{\p}{\partial}

\newcommand*{\one}{\mathbbm{1}}

\newcommand*{\FF}{\mathbb{F}}
\newcommand*{\CC}{\mathbb{C}}

\newcommand*{\ol}{\bar}   
\newcommand*{\sm}{\setminus}

\begin{document}

\maketitle
\begin{abstract}
In this paper, we present two constructions of QLTCs with constant soundness. In the first approach, we introduce an operation which we call check product, and show how this operation gives rise to QLTCs of constant soundness, constant rate, and distance scaling with locality. In the second approach, we utilize homological product of codes and prove a special case in which the soundness of component codes is preserved through the homological product. This observation leads us to construct QLTCs of constant soundness, scalable rate and distance, and constant average locality. Our work marks a step towards constructing QLTCs of high soundness and distance, which would give a different construction to the No Low-Energy Trivial States (NLTS) theorem.
\end{abstract}

\section{Introduction}
Quantum error correcting codes (QECCs) are essential objects of study in quantum information science due to their broad applications in both practical and theoretical quantum computation. Practically, QECC is the foundation of fault-tolerant quantum computation, which holds the promise to scalable quantum computing. Theoretically, quantum coding theory interacts extensive with quantum complexity theory and information theory, much like how classical coding theory interacts with theoretical computer science. In development of quantum codes for theoretical purposes, we often focus on their asymptotic properties, such as relative rate and distance. In this paper, we study the \emph{testability} of quantum codes and present two novel constructions.

A classical code $C$ is called testable if the syndrome of a proposed code word $w$ reveals more than whether $w$ belongs to the code: the relative weight of the syndrome is also proportional to the relative distance of $w$ from the codespace, and their ratio is called the \emph{soundness} of $C$. A code is further called \emph{locally testable} if all of its checks involves at most a constant number of bits from $w$. The theory of code checking, which began with the pioneering work of Blum, Luby, and Rubinfeld~\cite{blum1990self}, has grown into a widely successful area of the theory of computing, affecting PCP theory, combinatorial optimization, combinatorial property testing, program checking and even cryptography.

With the advancement of quantum information science, a quantum notion of locally testable codes (QLTC) was first proposed and studied by Aharonov and Eldar in~\cite{aharonov2015quantum}. The existence of such codes gained major interest in 2015, when Eldar and Harrow showed in~\cite{eldar2015local} that any QLTC with constant soundness, locality and relative distance could be used to construct Hamiltonians with no low-energy trivial states, which would resolve the famous NLTS conjecture
of Freedman and Hastings~\cite{freedman2013quantum}. This conjecture is deeply related to the quantum PCP conjecture~\cite{aharonov2013guest}, one of the most important open problems in quantum complexity theory.
Further, classical LTCs are important components in the proofs of the classical PCP Theorem~\cite{arora1998proof, dinur2007pcp}. Naturally, the connections between QLTCs and the qPCP conjecture became a topic of major importance in the field.

However, constructing QLTCs of high soundness and distance seemed far-fetched at the time, as the best known quantum LDPC codes (QLDPC) had distance $\tilde{\Theta}(\sqrt{N})$ and no bound on soundness. Here we say that a (quantum) code is LDPC if the number of (qu)bits each check involves 
and the number of checks each (qu)bit is involved in 
are all bounded by a constant $\ell$, which we call the \emph{locality} of the code. The first QLTC with unconditional guarantee on soundness is the hypersphere product code constructed by Hastings~\cite{hastings2016quantum}, which encodes two logical qubits, has soundness $1/\log(N)^2$, locality $\Theta(\log(N))$ and distance $\Theta(\sqrt{N})$. In 2019, Leverrier, Londe and Z\'{e}mor constructed another family of QLTC called the hemicubic codes~\cite{leverrier2022towards}, which encodes one qubit and has an improved soundness of $1/\log(N)$, with roughly the same locality and distance. These constructions remain the only known QLTCs, and it was unclear how to improve the soundness to $\Omega(1)$, without reducing the code to have no encoded qubits. 

The landscape of QLDPC constructions changed significantly in 2020, following Hastings, Haah and O'Donnell's breakthrough construction of fiber bundle codes~\cite{hastings2021fiber} that achieved $\tilde{\Omega}(N^{3/5})$ distance. Their ideas were soon generalized to lifted product codes~\cite{panteleev2021quantum} by Panteleev and Kalachev, and balanced product codes~\cite{breuckmann2021balanced} by Breuckmann and Eberhardt, culminating in the 2021 result of Panteleev and Kalachev that presented the first family of asymptotically good qLPDC codes~\cite{panteleev2022asymptotically}. It was further shown that embedded in the construction in~\cite{panteleev2022asymptotically} is a $c^3$-LTC: asymptotically good classical LDPC code (CLDPC) that is locally testable with constant soundness. Around the same time, Dinur, Evra, Livne, Lubotzky and Mozes announced their construction of a $c^3$-LTC using left-right Cayley complexes~\cite{dinur2022locally}. Building upon their works, two other families of asymptotically good QLDPCs are constructed~\cite{leverrier2022quantum, dinur2022good}. One such family, namely quantum Tanner codes~\cite{leverrier2022quantum} by Leverrier and Z\'emor, was utilized by Anshu, Breuckmann and Nirkhe~\cite{anshu2022nlts} to prove the NLTS conjecture. Following these developments, the problem of constructing QLTCs with good parameters became one of major interest. 


In this paper, we take a concrete step towards this open problem by constructing the first few families of QLTCs with constant soundness and varying distance.

\subsection{Main Results}

We begin by presenting a simple idea that transforms a good classical LDPC locally testable code with constant soundness into a quantum LDPC locally testable code with constant soundness and constant rate. However, it has the striking deficit of having distance 2. 
\begin{lemma}\label{thm:dist2}
Given a family of classical LDPC codes with parameters $[n, k, d]$ that are locally testable with soundness $\rho$, there exists a family of quantum LDPC codes that are locally testable with soundness $2\rho$ and parameters $[2n, 2k, 2]$. 
\end{lemma}
It is important to note that the soundness condition is nontrivial, because it concerns the distance of an \emph{arbitrary} state from the codespace, which can be much larger than the code distance, which is the distance \emph{between} two codewords. For instance, if $C$ is a classical LTC over $n$ bits, then the code $\bar{C} = \{(x, y) \in \mathbb{F}_2^{2n} : x + y \in C\}$ has distance 2 (since it contains $(e_i, e_i)$ for all $i$), yet if $x$ is far from $C$, then $(x, 0)$ is far from $\bar{C}$. Our quantum construction is in fact based on this simple classical example.
We present this construction in section~\ref{sec:cp} and discuss how it could be generalized to a new operation which we give the name \emph{check product}. Using this operation and random quantum codes, we obtain the first main result of our paper, a family of QLTCs with constant soundness, constant rate, and distance scaling with locality. 

\begin{theorem}\label{thm:scale}
Suppose we have a family of classical LDPC codes with parameters $[n_1, k_1 = rn_1, d_1]$ that are locally testable with soundness $\rho$. Then for any $n_2$, there exists a family of quantum locally testable codes with soundness $4\rho$, locality bounded by $O(n_2)$, and parameters $[n_1n_2, (1+r)n_1n_2/2, \Theta(\min(d_1,n_2))]$.
\end{theorem}

For instance, taking $n_2 = \log(n_1)$, we obtain QLTCs of constant soundness, constant rate, $\Theta(\log(n))$ distance and $O(\log(n))$ locality.

Our second main result utilizes the distance balancing technique for quantum codes introduced by Hastings~\cite{hastings2016weight}. Given a quantum code of distinct $X$ and $Z$ distance, Hastings showed that one could take the hypergraph product of this quantum code with a classical repetition code to obtain a new quantum code with balanced $X$ and $Z$ distance. 
Notably, this technique balances distance at the expense of soundness -- if the original quantum code has soundness $\rho$ and one uses a repetition code of length $\ell$ in the hypergraph product, the soundness of the resulting quantum code would be $\rho/\ell$. 

To address this deficiency, we propose a slight yet critical modification to the distance balancing technique such that the soundness of the component quantum code is preserved, at partial expense in locality. We present these results in Section~\ref{sec:distbalance}, arriving at our second main result.





\begin{theorem}\label{thm:newdistbala}
Fix integer $\ell \ge 2$. Given a family of classical LDPC codes with parameters $[n, k, d]$ that are locally testable with $m$ checks and soundness $\rho$, there exists a family of quantum locally testable codes of soundness $\Omega(\rho)$ and parameters $[N = n\ell + m(\ell-1), k, \min(d, 2\ell)]$, such that a $1/\ell$ fraction of $X$-stabilizer generators have weight $\Theta(\ell)$, and at most $1/\ell$ fraction of the qubits are checked by $\Theta(\ell)$ $Z$-stabilizer generators. All other stabilizer generators are constant weight, and all other qubits are checked by a constant number of stabilizer generators. 
\end{theorem}

As an example, choosing $\ell = n$, we obtain QLTCs with constant soundness, $\Theta(1/\sqrt{N})$ rate, $\Theta(\sqrt{N})$ distance and $O(\sqrt{N})$ locality. While this code family is no longer LDPC for any $\ell$ scaling with $n$, we note that the check weights are non-uniform. In particular, we note that the total check weight of all stabilizer generators is $\Theta(N)$, which means the average locality is $\Theta(1)$. While it is unclear whether average locality is an useful measure, we note the constant average locality here since it is an interesting property that arises naturally from our constructions.

We now present a summary of parameters of known QLTCs and our constructions in the following tables. Since some of these constructions have components whose size can be tweaked (such as the length of the repetition code in Theorem~\ref{thm:newdistbala}), in table~\ref{tab:general} we present the parameters of the general forms of these constructions, and in table~\ref{tab:special} we present the parameters of special cases for direct comparison.
We further remark that in terms of worst case bounds, the parameters of Theorem~\ref{thm:scale} are strictly better than that of Theorem~\ref{thm:newdistbala}. However, the average locality of Theorem~\ref{thm:newdistbala} are both constant, which is not the case in Theorem~\ref{thm:scale}.

\begin{table}[H]
    \centering
    \begin{tabular}{| c || c | c | c | c | c | c |}
    \hline
    Constructions & Ref~\cite{hastings2016quantum} & Ref~\cite{leverrier2022towards} & Theorem~\ref{thm:scale} & Theorem~\ref{thm:newdistbala} \\
    \hline 
    \hline 
    Physical qubits  
    & $N \approx (2p)^{2n}$ & $N$ & $N = n_1n_2$ & $N = nl$ \\
    \hline 
    Soundness 
    & $1/(np)$ & $\Omega(1/\log(N))$ & $\Omega(1)$ & $\Omega(1)$  \\
    \hline 
    Logical qubits 
    & 2 & 1 & $n_1n_2$ & $n$ \\
    \hline 
    Distance
    & $\Theta(p^{n+1})$ & $\Theta(\sqrt{N})$ & $\Theta(\min(n_1, n_2))$ & $\Theta(\min(n, \ell))$  \\
    \hline 
    Locality
    & $\Theta(2n)$ & $O(\log(N))$ & $O(n_2)$ &  $\avg= \Theta(1), \max = \Theta(\ell)$ \\
    \hline 
    \end{tabular}
    \caption{Parameters of the general forms of the constructions. }
    \label{tab:general}
\end{table}

\begin{table}[H]
    \centering
    \begin{tabular}{| c || c | c | c | c | c | c |}
    \hline
    Constructions & Ref~\cite{hastings2016quantum} & Ref~\cite{leverrier2022towards} & Theorem~\ref{thm:scale} & Theorem~\ref{thm:newdistbala} \\
    \hline 
    \hline 
    Physical qubits  
    & $N$ & $N$ & $N$ & $N$ \\
    \hline 
    Soundness 
    & $1/\log(N)^2$ & $\Omega(1/\log(N))$ & $\Omega(1)$ & $\Omega(1)$  \\
    \hline 
    Logical qubits 
    & 2 & 1 & $\Theta(N)$ & $\Theta(\sqrt{N})$ \\
    \hline 
    Distance
    & $\Theta(\sqrt{N})$ & $\Theta(\sqrt{N})$ & $\Theta(\log(N))$ & $\Theta(\sqrt{N})$ \\
    \hline 
    Locality
    & $O(\log(N))$ & $O(\log(N))$ & $O(\log(N))$ &  $\avg = \Theta(1),\max = \Theta(\sqrt{N})$ \\
    \hline 
    \end{tabular}
    \caption{Parameters of special cases of the constructions. }
    \label{tab:special}
\end{table}

\section{Preliminaries}

In this section, we introduce the basic definitions of quantum CSS codes and local testability. Given a vector $v\in \FF_2^n$, we let $X^v$ denote the $n$-qubit Pauli operator $X^{v_1}\otimes X^{v_2}\otimes \cdots X^{v_n}$, where $X^1 = X$ and $X^0 = I$. We define $Z^v$ similarly. 
\begin{definition}[CSS Codes]
Let $C_X=\ker(H_X)$ and $C_Z=\ker(H_Z)$ be linear codes of length $n$ such that $C_X^\perp\subseteq C_Z$ (equivalently, $H_XH_Z^T=0$).
The quantum code $Q = \css(H_X,H_Z)$ is the stabilizer code where the $X$-stabilizers have the form $X^c$ for $c\in C_X^{\perp}$, and the $Z$-stabilizers have the form $Z^c$ for $c\in C_Z^{\perp}$. Its code space is spanned by following states, where for $\gamma$ ranges over $C_Z$. 
\begin{equation}
|\gamma + C_X^\perp\rangle := \frac{1}{\sqrt{|C_X^\perp|}}\sum_{c\in C_X^\perp} |\gamma+c\rangle
\end{equation}
\end{definition}

\begin{fact}
If $C_X$ and $C_Z$ are $[n,k_X,d_X']$ and $[n,k_Z,d_Z']$ codes, respectively, then $Q=\css(H_X,H_Z)$ has dimension $k=k_X+k_Z-n$ and minimum distance $d=\min(d_X,d_Z)\geq \min(d_X',d_Z')$ where $d_X := \min \{ |c| : c\in C_X\setminus C_Z^\perp\}$ and $d_Z := \min \{ |c| : c\in C_Z\setminus C_X^\perp\}$.
\end{fact}

In this paper, we consider the following definition of local testability, which is the same as in~\cite{leverrier2022towards}.

\begin{definition}[Locally testable code]
A linear code $C\in \FF_2^n$ is locally testable with soundness $\rho$ and check weight $w$ if it has parity check matrix $H\in \FF_2^{m\times n}$ with rows of weight $w$ such that for any $x\in \FF_2^n$ we have
\begin{equation}
    \frac{1}{m}|Hx|\geq \frac{\rho}{n}d(x,C)
\end{equation}
where $d(x,C):=\min_{c\in C} d(x,c)$ and $d(\cdot,\cdot)$ denotes the Hamming distance.
\end{definition}

We consider the definition of quantum locally testable codes as in~\cite{eldar2015local}. Given a quantum stabilizer code with stabilizer generators $S_1, \cdots, S_m$ all having weight at most $w$, we define the projector operators $\Pi_i = \frac{1}{2}(I - S_i)$. For a quantum codespace $Q \le (\CC^2)^{\otimes n}$, we define the \emph{$t$-fattening} of $Q$ as
\[
Q_t = \Span{\{(A_1\otimes \cdots \otimes A_n)\ket{\psi}: \ket{\psi}\in Q, \#\{i\in [n], A_i\ne I\} \le t\}}.
\]
This is the space of states that are at distance at most $t$ from the codespace $Q$. Let $\Pi_{Q_t}$ be the projector onto $Q_t$, and let 
\[
D_Q = \sum_{t\ge 1}t(\Pi_{Q_t} - \Pi_{Q_{t-1}}).
\]

\begin{definition}[Quantum Locally Testable Codes]
A $n$-qubit quantum code $Q$ with stabilizer generators $S_1, \cdots, S_m$ is locally testable with soundness $\rho$ and check weight $w$ if all its stabilizer generators have weight at most $w$, and
\[
\frac{1}{m}\sum_{i=1}^m\frac{1}{2}(I - S_i) \succeq \frac{\rho}{n}D_Q.
\]
\end{definition}

Local testability of quantum CSS codes and the testability of their classical component codes are closely related, as shown by~\cite{eldar2015local}:
\begin{lemma}[Fact~17 of~\cite{eldar2015local}]~\label{lem:fact17}
A quantum CSS code $\css(H_X, H_Z)$ is a QLTC with soundness $\rho$ if $C_X = \ker(H_X), C_Z = \ker(H_Z)$ are CLTCs of soundness $\rho$. Conversely, if $\css(H_X, H_Z)$ is a QLTC with soundness $\rho$, then $C_X, C_Z$ are CLTCs of soundness at least $\rho/2$.
\end{lemma}
With this lemma in place, we now proceed to present our constructions.

\section{Check Product of Codes}\label{sec:cp}
We start by presenting a simple construction that proves Lemma~\ref{thm:dist2}. All proofs in section~\ref{sec:cp} and \ref{sec:cpofLTCs}, except for that of Claim~\ref{clm:dup}, are included in the appendix.

\subsection{A Motivating Example}\label{sec:motiv}
Suppose $C = \ker(H)$, $H\in \FF_2^{m\times n}$ is a classical LTC with soundness $\rho$, rate $r$, distance $d$, and in addition an LDPC with weight $w$. Define $\ol{C} = \ker(\ol{H})$, where $\ol{H} = [H, H]$. Then $\ol{H}\ol{H}^T = 0$, which means $Q = \css(\ol{H}, \ol{H})$ is a valid quantum code. 
We now prove the following claim, which justifies considering this approach.
\begin{claim}\label{clm:dup}
$\ol{C} = \ker(\ol{H})$ is a LDPC CLTC with soundness $2\rho$, rate $\frac{1+r}{2}$, and distance 2.
\end{claim}
\begin{proof}
A way to understand the code $\ol{C}$ is through its Tanner graph. Suppose the Tanner graph of $C$ consists of bit vertices $B$ and check vertices $K$. Then the Tanner graph of $\ol{C}$ is obtained simply by creating a copy $v'$ of each $v\in B$, where $v'$ and $v$ are connected to the same check bits in $K$. We can therefore represent each $z\in \ol{C}$ as $z = (x, y)$, where $x, y\in \FF_2^n$. We show
\begin{equation}\label{eq:newcode}
    \ol{C} = \{(x, y): x, y\in \FF_2^n, x+y\in C\}.
\end{equation}
Fix $(x, y)$ such that $x + y\in C$. Then $\ol{H}(x, y) = Hx + Hy = 0$, which means $(x, y)\in \ol{C}$. Similarly, fix $(x, y)\in \ol{C}$, then we have $H(x + y) = \ol{H}(x,y) = 0$. This proves~\eqref{eq:newcode}. Now we see that $\ol{C}$ has distance $2$, because for any $v\in B$, let $\one_v\in \FF_2^n$ be the indicator vector of $v$ (meaning that it has a one at index $v$, and 0 elsewhere), then $(\one_v, \one_{v'})\in \ol{C}$. We observe that any weight $1$ vector $v\in \FF_2^{2n}$ will have non-zero syndrome.

Now note that the check weights of $\ol{C}$ are bounded by $2w$, and each qubit is checked by at most $w$ checks. Therefore $\ol{C}$ is LDPC. The rate of $\ol{C}$ is also easy to compute -- the number of linearly independent checks stay at $(1-r)n$, while the number of bits is doubled. So the overall rate is $\frac{1+r}{2}$.

The more interesting part is to show local testability. We want to show $\forall x, y\in \FF_2^n$,
\[
|\ol{H}(x, y)|/m \ge 2\rho \cdot d((x, y), \ol{C})/2n. 
\]
We first note that $|\ol{H}(x, y)| = |Hx + Hy| = |H(x+y)|$. Moreover, $d((x, y), \ol{C}) = d(x+y, C)$. By local testability of $C$, we have
\[
|H(x + y)|/m \ge \rho \cdot d((x+y), C)/n,
\]
which completes our proof.
\end{proof}
\begin{restatable}{corollary}{CoroDistTwo}
The quantum code $Q = \css(\ol{H}, \ol{H})$ is a QLTC of soundness $2\rho$, check-weights bounded by $2w$, rate $r$, and $d_x = d_z = 2$.
\end{restatable}


By choosing $C$ as a known $c^3$-LTC, such as the left-right Cayley complex code~\cite{dinur2022locally}, we obtain Lemma~\ref{thm:dist2}. It is surprising that such a simple construction could already give us quantum codes of constant soundness. Moreover, this example demonstrates that the soundness and distance of a quantum code are not necessarily related. 
We generalize this construction in the following section.

\subsection{General Check Products}
We begin by making the following observation: the matrix $\ol{H} = [H, H]$ of the previous section could also be written as $[1, 1]\otimes H$, where $H_0 = [1, 1]$ is the parity check matrix of the repetition code $C_0 = \{00,11\}$. More generally, for any two classical codes $C_1 = \ker(H_1), C_2 = \ker(H_2)$, we can define the following \emph{check product} of the two codes:
\[
C_1\star C_2 = \ker(H_1\otimes H_2).
\]
The following facts are well known.

\begin{restatable}{lemma}{CheckProductProperty}\label{lem:cpprop}
Let $C_1, C_2$ be classical codes with parameters $[n_1, k_1, d_1]$ and $[n_2, k_2, d_2]$. The following hold:
\begin{enumerate}
    \item $C_1\star C_2 = C_1\otimes \FF_2^{n_2} + \FF_2^{n_1} \otimes C_2$. Namely, $C_1\star C_2$ is the dual tensor code of $C_1$ and $C_2$. 
    \item $C_1\star C_2$ has dimension $n_1n_2 - (n_1-k_1)(n_2-k_2)$. 
    \item $C_1\star C_2$ has distance $\min(d_1, d_2)$.
\end{enumerate}
\end{restatable}

With this definition formalized, we can extend it to the check product of a classical code and a quantum code, as follows.
\begin{definition}[Check-Product of classical and quantum code]
Given a classical code $C = \ker(H)$ and a quantum CSS code given by $Q = \css(H_X, H_Z)$, we define the check-product of $Q$ and $C$ to be the quantum code $Q\star C = \css(H_X\otimes H, H_Z\otimes H)$.  
\end{definition}
It is straightforward to see that the commutativity condition is satisfied, so this definition is valid. We now address the distance of this quantum code.

\begin{restatable}{lemma}{CPDistance}\label{lem:cpquantum}
Given a quantum code $Q = \css(H_X, H_Z)$, the code $Q\star C$ has distance $\min(d(C), d(C_X), d(C_Z))$. In words, its distance is the minimum of all of its component classical codes' distance.
\end{restatable}

The natural question to ask now is --- what can we say about the soundness of general check product codes? It is straightforward to see that the check product of a CLTC with a classical code that is not locally testable is also not locally testable, and it is also not clear if the check product of two arbitrary CLTCs remains locally testable. However, as we will show in section~\ref{sec:cpofLTCs}, the check product of CLTCs with a specific form is indeed locally testable. This enables us to prove Theorem~\ref{thm:scale} by considering the check product of a $c^3$-LTC with random quantum CSS codes.

\section{The Check Product of a CLTC and a QLTC}\label{sec:cpofLTCs}
In this section, we present a construction that proves Theorem~\ref{thm:scale}. We once again begin by making a simple observation. Given a classical code $C = \ker(H)$, suppose $H$ is a $m\times n$ matrix with linearly independent rows. We may do Gaussian operations on the rows of $H$, formally multiplying $H$ from the left by a non-singular matrix $G$, such that the resulting matrix has the form (up to permuting the columns with a permutation matrix, $\Pi$):
\begin{equation}~\label{eq:LTCform}
H' = \begin{bmatrix}I_{m} | R\end{bmatrix} \;\; (= GH\Pi)   
\end{equation}
where $I_{m}$ is the $m\times m$ identity matrix, and $R$ is a $m\times (n-m)$ matrix. Note that $\ker(H') = \Pi^{-1} \ker(H) = \Pi^{-1} C$, and the rows of $H'$ can have arbitrary weight due to the row operations. 

\begin{restatable}{claim}{SimpleForm}\label{cl:tr}
The code $C$ with check matrix $H'\Pi^{-1}$ has soundness $\ge 1/r$, where 
$r$ is the rate of $C$. 
\end{restatable}

This claim directly implies the following simple corollary, which we include without proof.
\begin{corollary}~\label{cor:transform}
Any classical linear code of rate $r$ can be turned into a CLTC with soundness $\ge 1/r$ at the cost of having arbitrary locality, while keeping the same rate and distance. Similarly, any quantum CSS codes where both component codes are classical linear codes can be turned into a QLTC with soundness 
$\ge 1/r$ at the cost of having arbitrary locality, while keeping the same rate and distance.
\end{corollary}

We note that the same idea of Claim~\ref{cl:tr} was discussed by Campbell in section~5 of~\cite{campbell2019theory}, but we only became aware of this correspondence after this paper was finished.

In spite of its simplicity, Claim~\ref{cl:tr} has an important 
application in our check product construction, as shown in the following lemma.

\begin{restatable}{lemma}{CPTestability}~\label{lem:cpCLTC}
Suppose $C = \ker(H)\subset \FF_2^n$ is a CLTC with soundness $\rho$ and locality $w$. Let $C_X = \ker(H_X)$ be a classical code where $H_X$ is a $m_X\times n_X$ matrix of the form $[I_{m_X}\mid h_X]$, such that its locality is bounded by $w_X$. Then $C_X\star C$ is a CLTC with soundness at least 
$\rho n_X/m_X$ and locality $ww_X$.
\end{restatable}

Combining Lemma \ref{lem:cpCLTC}
with Lemma \ref{lem:cpprop} and Corollary~\ref{cor:transform}, we obtain:
\begin{restatable}{theorem}{CPQLTC}~\label{thm:cp2QLTC}
Given a $n$-bit classical LTC, $\,C = \ker(H)$, of soundness $\rho$, dimension $k$, distance $d$ and locality $w$, and a $n_q$-qubit quantum code $Q = \css(H_X, H_Z)$ of dimension $k_q$, we can reduce $H_X, H_Z$ to have the form in equation~\eqref{eq:LTCform}, and denote the resulting quantum code $\ol{Q} = \css(\ol{H}_X, \ol{H}_Z)$. ($Q$ and $\ol{Q}$ as sub-spaces are the same. What makes them different is their sets of $X$ and $Z$ checks.) Then the check product code $\ol{Q}\star C$ has
\begin{enumerate}
    \item local testability with soundness $\rho\cdot\min(\frac{n_q}{m_X}, \frac{n_q}{m_Z})$ (in 
    particular, soundness $\ge \rho)$,
    \item distance $\min(d, d(C_X), d(C_Z))$, where $C_X = \ker(H_X), C_Z = \ker(H_Z)$.
    \item locality bounded by $wn_q$, 
    \item and dimension $nn_q - (n-k)(n_q-k_q)$.
\end{enumerate}
\end{restatable}

While it is tempting to apply this theorem to a $c^3$-LTC and a good QLDPC code, we note that the result will have constant distance. Indeed, for a quantum CSS code that is LDPC, we necessarily have $d(C_X), d(C_Z) = O(1)$. For instance, to argue
that $d(C_X) = O(1)$,
we notice that $C_X\ge C_Z^{\perp}$, 
and since $C_Z^{\perp}$ contains all possible checks for $C_Z$, it also contains the low-weight ones. In fact, to construct QLTC of scalable distance from Theorem \ref{thm:cp2QLTC}, we want $d(C_X), d(C_Z)$ to be as large as possible, which in turn implies that the quantum code $Q = \css(H_X, H_Z)$ has correspondingly large check weights for all possible checks (i.e. checks in $C_X^{\perp}$ and
$C_Z^{\perp}$). In the following theorem, we show that random codes satisfy this property.

\begin{restatable}{theorem}{RandomQuantumCodes}~\label{thm:dualcode}
Let $C\le \FF_2^n$, $C = \ker(H_C)$ be a random code with dimension $\frac{3n}{4}$. Let $D\le C$ be a random subspace of $C$ of dimension $\frac{n}{4}$. Let $H_D$ be the $\frac{n}{4}\times n$ matrix such that its rows span the space $D$, then with high probability $C$ and $D^{\perp}$ both have linear distances, and the {\rm CSS}
code made from
the classical codes $C$ and $D^{\perp}$ 
(here check sets do not influence the statement) has rate $1/2$.  
\end{restatable}

Combining Theorem~\ref{thm:cp2QLTC} and~\ref{thm:dualcode}, we obtain Theorem~\ref{thm:scale}. 


\section{Gauge Fixing and Distance Balancing}\label{sec:distbalance}
Theorem~\ref{thm:scale} shows a family of QLTCs where the quantum distance scales positively with the check weight. In this section, we prove Theorem~\ref{thm:newdistbala} by tweaking our construction in section~\ref{sec:motiv} and applying distance balancing with our modifications.

\subsection{Modifying Stabilizer and Logical Operators}
We briefly recall our earlier construction. Given a LDPC CLTC $C = \ker{H}\le \FF_2^{n}$ of soundness $\rho$, dimension $k$ and distance $d$, we define $\ol{C} = \ker(\ol{H})$, where $\ol{H} = [H, H]$, and $Q = \css(\ol{H}, \ol{H})$. Let us consider the stabilizers and logical operators of $Q$. Consider the set of Pauli operators 
\[
\{X^{(e_i, e_i)}, i\in [n]: e_i\text{ are the standard basis vectors of $\FF_2^n$}\}.
\]
We observe that exactly $k$ of these operators are independent $X$-logical operators, and the other $n-k$ of them can be generated from the previous $k$ operators together with the $X$-stabilizers. There are another set of $X$-logical operators, namely
\[
\{X^{(v, 0)}: v\in C\}.
\]
Exactly $k$ logical operators in this set are independent, and they all have linear weight. Together, these $2k$ operators generate the complete set of $X$-logical operators of $Q$. Similarly, by replacing $X$s with $Z$s, we found the set of $Z$-logical operators of $Q$. We see that $d_x = d_z = 2$. 

Now suppose we move all the $Z$-logical operators of weight 2 into the $Z$-stabilizer group. Then the remaining $k$ $Z$-logical operators all have linear weight, which means we have $d_z = O(n)$. On the other hand, all the high-weight $X$-logical operators of the form $\{X^{(v, 0)}: v\in C\}$ are no longer valid logical operators, as they do not compute with all the $Z$-stabilizers. Therefore, the set of $X$-logical operators that remain is precisely
\[
\{X^{(v, v)}: v\in \FF_2^n\}.
\]

A more direct way of writing this new code would be the following. Let $H_Z = [I, I], H_X = [H, H]$. Our new code is precisely $Q' = \css(H_X, H_Z)$. The code states of $Q'$ can be explicitly written out as
\[
\ket{\psi_{v}} = \sum_{\substack{(a,b)\in \FF_2^n \\ a + b\in \ker{H}}}\ket{v+a, v+b}.
\]
This code then has $d_z = O(n), d_x = 2$. It is locally testable with soundness $2\rho$, LDPC, and has dimension $k$.

\subsection{Distance Balancing}\label{sec:constsound}
We can now apply the distance balancing techniques in~\cite{hastings2016weight}. Specifically, we consider our code $Q'$ as a chain complex 
\[
Q = \FF_2^{n}\xrightarrow{H_Z^T = [I, I]^T} \FF_2^{2n} \xrightarrow{H_X = [H, H]} \FF_2^m,
\]
and we take a repetition code of length $\ell$, also viewed as a chain complex
\[
R = E \xrightarrow{H_{\ell}} V,
\]
where $E = \FF_2^{\ell-1}, V = \FF_2^{\ell}$, and $H_{\ell}$ is the $\ell\times (\ell-1)$ matrix of the form
\begin{equation}~\label{eq:oldmtx}
H_{\ell} = 
\begin{bmatrix}
1 & 0 & 0 & 0 & \cdots  \\
1 & 1 & 0 & 0 & \cdots  \\
0 & 1 & 1 & 0 & \cdots  \\
& &  \cdots & & \\
& &  \cdots & & \\
0 & 0 & \cdots & 1 & 1 \\
0 & 0 & \cdots & 0 & 1 \\
\end{bmatrix}.
\end{equation}

We take the homological product of the two complexes, and the resulting chain complex $Q\times R$ is 
\begin{equation}~\label{fig:complex}
\begin{tikzcd}
    & \FF_2^{2n}\otimes E \ar[r, "H_X\otimes I"] \ar[ddr, "I\otimes H_{\ell}"]
        & \FF_2^m\otimes E \ar[dr, "I\otimes H_{\ell}"] \\
\FF_2^n\otimes E  \ar[ur, "H_Z^T\otimes I"] \ar[dr, "I\otimes H_{\ell}"'] 
    & \oplus
        & \oplus
            & \FF_2^m\otimes V \\
    & \FF_2^{n}\otimes V  \ar[r, "H_Z^T\otimes I"]
        & \FF_2^{2n}\otimes V \ar[ur, "H_X\otimes I"'] \\
C_3 \ar[r, "\partial_3"]
    & C_2 \ar[r, "\partial_2"]
        & C_1 \ar[r, "\partial_1"]
            & C_0
\end{tikzcd}
\end{equation}

Now we take the sub-chain complex $C_2\rightarrow C_1\rightarrow C_0$, and view it as a quantum code. From~\cite{hastings2016weight}, the following holds.
\begin{lemma}~\label{lem:oldpara}
If the original quantum code $Q$ is a $2n$-qubit LDPC quantum code with dimension $k$ and distance $d_x, d_z$, then this new code $Q'$ is a $2n\ell+m(\ell-1)$-qubit LDPC quantum code with dimension $k$ and distance $\ell d_x, d_z$.
\end{lemma}

We refer the readers to~\cite{hastings2016weight}, statement~7 and~8 of Lemma~2 for the proofs. Moreover, it was also shown in~\cite{hastings2016weight} that the resulting code is locally testable, albeit having a lower soundness.
\begin{lemma}[Lemma~7 of~\cite{hastings2016weight}]~\label{lem:oldsound}
If $Q$ is a QLTC with constant soundness, then the code $Q'$ is a QLTC with soundness $\Omega(1/\ell)$. 
\end{lemma}
Combining these two lemmas, we obtain the following corollary. 

\begin{corollary}\label{thm:distbala}
Fix integer $\ell \ge 2$. Given a family of classical LDPC codes with paramaters $[n, k, d]$ that are locally testable with $m$ checks and soundness $\rho$, there exists a family of quantum LDPC locally testable codes of soundness $\Omega(\rho/\ell)$ and parameters $[N= n\ell + m(\ell-1), k, \min(d, 2\ell)]$.
\end{corollary}

We choose not to include direct proofs of these two lemmas in this paper, because in the following section we would present a slight modification to the above construction and prove Lemma~\ref{lem:newpara} and~\ref{lem:newsound}. The proof of Lemma~\ref{lem:newpara} would be a slightly modified version of the proof of Lemma~\ref{lem:oldpara} in~\cite{hastings2016weight}, and the proof of Lemma~\ref{lem:newsound} would follow a similar scheme as the proof of Lemma~\ref{lem:oldsound} in~\cite{hastings2016weight}.


\subsection{Preserving Soundness}
As the primary goal of this paper is to achieve constant soundness, we present a simple modification to the above distance balancing technique that preserves soundness, while sacrificing some locality. We perform column operation on $H_{\ell}$ such that it has the following form.
\begin{equation}~\label{eq:newmtx}
H_{\ell} = 
\begin{bmatrix}
1 & 0 & 0 & 0 & \cdots  \\
0 & 1 & 0 & 0 & \cdots  \\
0 & 0 & 1 & 0 & \cdots  \\
& & \cdots & & \\
& & \cdots & & \\
0 & 0 & \cdots & 0 & 1 \\
1 & 1 & \cdots & 1 & 1 \\
\end{bmatrix}.
\end{equation}

Note that as before, $H_{\ell}^T$ is a valid parity check matrix for the repetition code. The rest of the construction remains unchanged.

While this modification seems superficial on first glance, it actually follows important geometric insights. To see that, we construct two graphs from the two matrices in equation~\eqref{eq:oldmtx} and~\eqref{eq:newmtx}. We create a vertex for each bit of the repetition code, which corresponds to rows of $H_l$, and for each column in $H_l$, we create an edge that connects the two vertices that has entry $1$ at that column. Then the previous matrix~\eqref{eq:oldmtx} gives a line segment, while the new matrix~\eqref{eq:newmtx} gives a star graph. Intuitively, the star graph has better soundness than a line segment since any collection of edges $S$ in the star graph must have at least the same number of vertices on its boundary $\partial S$, while in a line segment a partial line segment have only two vertices on its boundary. This intuition can be formally formulated as boundary and co-boundary expansion of chain complexes, which are directly related to the soundness of classical and quantum codes derived from such chain complexes. 

We now prove the following lemmas regarding the parameters of $Q'$. We suggest that the readers follow the chain complex in equation~\eqref{fig:complex} when reading the following lemmas and proofs.

\begin{lemma}~\label{lem:newpara}
$Q'$ is a $2n\ell+m(\ell-1)$-qubit quantum code with dimension $k$ and distance $\ell d_x, d_z$. A $1/\ell$ fraction of $X$-stabilizer generators have weight $\Theta(\ell)$, and at most $1/\ell$ fraction of the qubits are checked by $\Theta(\ell)$ $Z$-checks. All other checks are constant weight, and all other qubits are checked by a constant number of checks. 
\end{lemma}
\begin{proof}
To show the dimension of the new code, we cite the K\"{u}nneth formula from algebraic topology. Define the $r$th homology group of a chain complex as $H_r(C) = \ker{\partial_r}/\im{\p_{r+1}}$, then the number of logical qubits of a quantum code is precisely $\dim(H_1(C))$. From the K\"{u}nneth formula, we have
\[
H_1(Q\times R) = (H_1(Q)\otimes H_0(R)) \oplus (H_0(Q)\otimes H_1(R)).
\]
Note that $\dim(H_1(Q)) = k, \dim(H_0(R)) = 1,$ and $\dim(H_1(R)) = 0$. Therefore $\dim(H_1(Q\times R)) = k$. 

For the $Z$ distance of $Q'$, let $c = (x, y)\in C_1$ such that $x\in \FF_2^{2n}\otimes V$, $y\in \FF_2^m\otimes E$, and $Z^{(x, y)}$ is a $Z$-logical operator. Let $v_1, \cdots, v_{\ell}$ be the standard basis vector of $V = \FF_2^{\ell}$, and $e_1, \cdots, e_{\ell-1}$ be the standard basis vector of $E = \FF_2^{\ell-1}$. Write $x = \sum_{i=1}^{\ell} x_i\otimes v_i$. We describe a simplification procedure that turns $x$ into the form $\ol{x}\otimes v_1$, such that $|\ol{x}|\le \sum_{i=1}^{\ell}|x_i|$. 

We begin with $x_{\ell}\otimes v_{\ell}$. Consider $\p_2(x_{\ell}\otimes e_{1}) = (x_{\ell}\otimes (v_1+v_{\ell}), (H_Xx_{\ell})\otimes e_1)$, then 
\[
(\sum_{i=1}^{\ell} x_i\otimes v_i, y) + \p_2(x_{\ell}\otimes e_{1}) = ((x_{1}+x_{\ell})\otimes v_{1} + \sum_{i=2}^{\ell-1} x_i\otimes v_i, y + (H_Xx_{\ell})\otimes e_1).
\]
In the stabilizer formalism, this step correspond to multiplying the logical operator $Z^{(x, y)}$ with the $Z$-stabilizer $Z^{\p_2(x_{\ell}\otimes e_{1})}$. Now for all other $i = 2, \cdots, \ell-1$, we multiply by the $Z$-stabilizer specified by $\p_2(x_i\otimes (e_{1} + e_i)) = (x_i\otimes (v_1+v_i), (H_Xx_i)\otimes (e_{1} + e_i))$ and the final vector in $C_1$ will have the form
\[
\ol{c} = ((\sum_{i=1}^{\ell} x_i)\otimes v_1, \ol{y})
\]
for some $\ol{y}\in \FF_2^n\otimes E$. Since $Z^{\ol{c}}$ is a valid logical operator, we must have $\p_1(\ol{c}) = 0$, which means
\[
H_X(\sum_{i=1}^{\ell} x_i)\otimes v_1 + (I\otimes H_{\ell})\ol{y} = 0.
\]
However, note that $v_1\notin \im(H_{\ell})$, which means for the above equation to hold we must have $\ol{y} = 0$ and $\sum_{i=1}^{\ell} x_i \in \ker{H_X}$. This implies that $Z^{\sum_{i=1}^{\ell} x_i}$ must be a $Z$-logical operator of $Q$. Therefore
\[
|c| = |x|+|y| \ge |\sum_{i=1}^{\ell} x_i| + |y| \ge d_z(Q),
\]
which means $d_z(Q') \ge d_z(Q)$. We also note that $d_z(Q') \le d_z(Q)$ since if $Z^x$ is a $Z$-logical operator of $Q$, then $Z^{x\otimes v_1}$ is a $Z$-logical operator of $Q'$. We conclude that $d_z(Q') = d_z(Q)$.

To discuss the $X$ distance, we define cohomologies and cite the K\"{u}nneth formula again. For a chain complex $C$, define the $r$th homology group of a chain complex as $H^r(C) = \ker{\partial_{r+1}^T}/\im{\p_{r}^T}$, and the K\"{u}nneth formula in this case states
\[
H^1(Q\times R) = (H^1(Q)\otimes H^0(R)) \oplus (H^0(Q)\otimes H^1(R)).
\]
Since $H^1(R) = 0$, we have $H^1(Q\times R) = (H^1(Q)\otimes H^0(R))$, which means any $X$-logical operator (which corresponds to a chain $c$ in $H^1(Q\times R)$) can be written in the form 
\[
c = x\otimes v + \p_1^T(u)
\]
where $x\in H^1(Q), v\in H^0(R),$ and $u\in \FF_2^m\otimes V$. Note that the $H^0(R)$ has dimension $1$, which means $v = \sum_{i=1}^{\ell} v_i$. Now suppose $u = \sum_{i=1}^{\ell} u_i\otimes v_i$, and consider the projection of $c$ onto the space $\FF_2^{2n}\otimes V$. It has the form
\[
c\mid_{\FF_2^{2n}\otimes V} = \sum_{i=1}^{\ell} (x + H_X^Tu_i)\otimes v_i.
\]
Since $\p_2^Tc = 0$, we must have $x + H_X^Tu_i\in H^1(Q)$ for all $i$, which means $|c|\ge \ell d_x(Q)$. This shows $d_x(Q') \ge \ell d_x(Q)$. We also note that $d_x(Q') \le \ell d_x(Q)$ since if $X^x$ is a $X$-logical operator of $Q$, then $X^{x\otimes (\sum_{i=1}^{\ell} v_i)}$ is a $X$-logical operator of $Q'$. We conclude that $d_x(Q') = \ell d_x(Q)$.

For the locality of $Q'$, we note that $H_X\otimes I$ and $H_Z^T\otimes I$ are matrices with constant row and column weights. For $I\otimes H_{\ell}$, exactly $1/\ell$ fraction of its rows have weight $\ell$, while the other rows have weight $1$, and all its columns have weight $2$. Let us enumerate the standard basis vector of $\FF_2^{2n}$ as $q_1, \cdots, q_{2n}$, and the standard basis vectors of $\FF_2^m$ as $c_1, \cdots, c_m$. Then we have
\begin{enumerate}
    \item All the $Z$-stabilizer generators corresponding to basis vectors of $\FF_2^{n}\otimes V$ have constant weight;
    
    \item All the $Z$-stabilizer generators corresponding to basis vectors of $\FF_2^{2n}\otimes E$ have constant weight, because $I\otimes H_{\ell}^T$ has constant row weight;
    
    \item The qubits corresponding to $\FF_2^{2n}\otimes V$ have the form $q_i\otimes v_j$ for $i\in[2n], j\in[\ell]$. Qubits of the form $q_i\otimes v_{\ell}$ are checked by $\ell$ $Z$-stabilizer generators from $\FF_2^{2n}\otimes E$. All other qubits have constant degree (in both $X$ and $Z$ stabilizer generator checks).
    
    \item The $X$-stabilizer generators corresponding to $\FF_2^m\otimes V$ has the form $c_i\otimes v_j$ for $i\in[m], j\in[\ell]$. Generators of the form $c_i\otimes v_{\ell}$ checks $\ell$ qubits from $\FF_2^m\otimes E$. All other $X$-stabilizer generators have constant weight.
\end{enumerate}
In short summary, at most $1/\ell$ fraction of the qubits are checked by $\Theta(\ell)$ $Z$-stabilizer generators, and exactly $1/\ell$ fraction of the $X$-stabilizer generators have $\Theta(\ell)$ weight. All other check weights and qubit degrees are constant.
\end{proof}

Finally, we prove a lower bound on soundness. Recall that our code $Q$ has soundness $2\rho$.

\begin{lemma}~\label{lem:newsound}
$Q'$ is locally testable with soundness $\min(\rho,1)/8$ for $Z$-operators, and soundness $1/3$ for $X$-operators. 
\end{lemma}
\begin{proof}
Given a $Z$-operator $Z^c$ where $c\in C_1$, we once again turn $c$ into the following form by multiplying with $Z$-stabilizers:
\[
\ol{c} = (\ol{x}\otimes v_1, \ol{y}).
\]
Let $z\in \FF_2^{2n}$ be such that $|z| = d(\ol{x}, \ker{H_X})$ and $\ol{x}+z\in \ker{H_X}$. Then we see that 
\[
\p_1(\ol{c}) = \p_1((z\otimes v_1, \ol{y})).
\]
Therefore, it suffices for us to show that 
\begin{equation}~\label{eq:zsoundness}
\frac{|\p_1(\ol{c})|}{m\ell} \ge \frac{\min(\rho,1)\cdot(|z| + |\ol{y}|)}{8(2n\ell + m(\ell-1))}. 
\end{equation}

Write $\ol{y} = \sum_{i=1}^{\ell-1}y_i\otimes e_i$. Then we have
\begin{align*}
\p_1(\ol{c}) 
&= (H_X\ol{x})\otimes v_1 + \sum_{i=1}^{\ell-1}y_i\otimes (H_{\ell}e_i) \\
&= (H_X\ol{x})\otimes v_1 + \sum_{i=1}^{\ell-1}y_i\otimes v_i + (\sum_{i=1}^{\ell-1}y_i)\otimes v_{\ell} \\
&= (H_X\ol{x})\otimes v_1 + y_1\otimes v_1 + \sum_{i=2}^{\ell-1}y_i\otimes v_i + y_1\otimes v_{\ell} +  (\sum_{i=2}^{\ell-1}y_i)\otimes v_{\ell}.
\end{align*}
Let $y' = \sum_{i=2}^{\ell-1}y_i$. For two vectors $u, v,$ let $u\cap v$ denote the vector where $(u\cap v)(i) = 1$ if and only if $u(i) = v(i) = 1$. Then we have
\begin{align*}
|\p_1(\ol{c})|
&= |H_X\ol{x}| + |y_1| - 2|H_X\ol{x}\cap y_1| + \sum_{i=2}^{\ell-1}|y_i| + |y_1| + |y'| - 2|y'\cap y_1| \\
&= |H_X\ol{x}| - 2|H_X\ol{x}\cap y_1| + |y'| + 2|y_1| - 2|y'\cap y_1| + \sum_{i=2}^{\ell-1}|y_i|.
\end{align*}
We first observe that since $|u|, |v| \ge |u\cap v|$ for any $u,v$, 
\begin{align*}
|\p_1(\ol{c})|
&\ge \sum_{i=2}^{\ell-1}|y_i| + |y_1| - |y'\cap y_1|, \\
|\p_1(\ol{c})|
&\ge |H_X\ol{x}|.
\end{align*}
The first inequality implies $|\p_1(\ol{c})| \ge |\ol{y}/3|$, for the following reason. If $|y_1| \le 2|\ol{y}|/3$, then $|\p_1(\ol{c})| \ge \sum_{i=2}^{\ell-1}|y_i| \ge |\ol{y}|/3$; otherwise if $|y_1| \ge 2|\ol{y}|/3$, then $|y'| < \sum_{i=2}^{\ell-1}|y_i| < |\ol{y}|/3$, which again implies $|\p_1(\ol{c})| \ge |\ol{y}|/3$. 

Therefore, if $|H_X\ol{x}| \le |\ol{y}|/3$, then by soundness of $H_X$, 
\begin{align*}
|\p_1(\ol{c})|
&\ge |\ol{y}|/3 \ge \frac{|\ol{y}| + |H_X\ol{x}|}{4} \ge \frac{|\ol{y}|}{4} + \frac{2\rho md(\ol{x}, \ker{H_X})}{4\cdot 2n}, \\
\frac{|\p_1(\ol{c})|}{m\ell} 
&\ge \frac{|\ol{y}|}{4m\ell} + \frac{\rho d(\ol{x}, \ker{H_X})}{4n\ell}.
\end{align*}
Assuming $8(\ell-1)\ge 4\ell$, we have equation~\eqref{eq:zsoundness}. Now if $|H_X\ol{x}| \ge |\ol{y}|/3$, we also have 
\[
|\p_1(\ol{c})|\ge |H_X\ol{x}| \ge \frac{|\ol{y}| + |H_X\ol{x}|}{4},
\]
and the same calculations as above follows. Therefore, $Q'$ is locally testable with soundness $\rho/8$ for $Z$-operators.

For $X$-operators $X^c$ where $c\in C_1$, suppose $c = (x, \sum_{j=1}^{\ell-1}y_j\otimes e_j)$. Note that since $\im{H_{\ell}^T} = \FF_2^{\ell-1} = E$, we can multiply $X^c$ by $X$-stabilizers of the form $X^r$, where
\[
r = \p_1^T (\sum_{j=1}^{\ell-1}y_j\otimes v_j) = (\sum_{j=1}^{\ell-1}H_X^Ty_j\otimes v_j, \sum_{j=1}^{\ell-1}y_j\otimes e_j). 
\]
Then $c + r = (\ol{x}, 0)$ for some $\ol{x}$. In other words, we may assume without loss of generality that all $X$-operators $X^c$ has the form 
\[
c = (\sum_{i=1}^{\ell} x_i\otimes v_i, 0).
\]
We will show that $|\p_2^T(c)| \ge (\sum_{i=1}^{\ell} |x_i|)/2$. To prove this, we need to use the fact that our $H_Z$ has the form $[I, I]$. For notation purposes, we let
\[
x_i = (x_i^1, x_i^2) = (x_i^1(1), \cdots, x_i^1(n), x_i^2(1), \cdots, x_i^2(n)) \in \FF_2^{2n}.
\]

Further, without loss of generality, we may assume $x_{\ell}^1\cap x_{\ell}^2 = 0$ because if $x_{\ell}^1\cap x_{\ell}^2 \ne 0$, we can add $((x_{\ell}^1\cap x_{\ell}^2, x_{\ell}^1\cap x_{\ell}^2) \otimes (\sum_{i=1}^{\ell} v_i), 0)$ to $c$. Note here that $\p_2^T((x_{\ell}^1\cap x_{\ell}^2, x_{\ell}^1\cap x_{\ell}^2) \otimes (\sum_{i=1}^{\ell} v_i), 0) = 0$. Now we are ready for the proof. 

For clarity purposes, in the following equations, we denote ``addition mod 2'' as $+_m$. We have
\begin{align*}
\p_2^T(c) 
&= \sum_{i=1}^{\ell}(x_i^1 +_m x_i^2)\otimes v_i +_m \sum_{i=1}^{\ell}(x_i^1, x_i^2)\otimes (H_{\ell}^T v_i) \\
&= \sum_{i=1}^{\ell}(x_i^1 +_m x_i^2)\otimes v_i +_m \sum_{i=1}^{\ell-1}(x_i^1, x_i^2)\otimes e_i +_m (x_{\ell}^1, x_{\ell}^2)\otimes (\sum_{i=1}^{\ell-1} e_i). 
\end{align*}
\begin{equation*}\label{eq:synw}
\begin{aligned}
|\p_2^T(c)|
&= \sum_{i=1}^{\ell}|x_i^1 +_m x_i^2| + \sum_{i=1}^{\ell-1}(|x_i^1| + |x_i^2|) + \sum_{i=1}^{\ell-1}(|x_{\ell}^1| + |x_{\ell}^2|) - 2\sum_{i=1}^{\ell-1} (|x_i^1\cap x_{\ell}^1| + |x_i^2\cap x_{\ell}^2|), \\
&= |x_{\ell}^1 +_m x_{\ell}^2| + \sum_{i=1}^{\ell - 1}\sum_{j=1}^n [ (x_i^1(j) +_m x_i^2(j)) + x_i^1(j) + x_i^2(j) + x_{\ell}^1(j) + x_{\ell}^2(j) \\
&\phantom{{} = |x_{\ell}^1 +_m x_{\ell}^2| + \sum\sum [}
- 2(x_i^1\cap x_{\ell}^1)(j) - 2(x_i^2\cap x_{\ell}^2)(j) ]. 
\end{aligned}  
\end{equation*}

We abbreviate the above equation into
\[
|\p_2^T(c)| = |x_{\ell}^1 +_m x_{\ell}^2| + \sum_{i=1}^{\ell - 1}\sum_{j=1}^n c_{i,j}.
\]
We now prove $|\p_2^T(c)| \ge \frac{1}{2}\sum_{i=1}^{\ell}\sum_{j=1}^n |x_i^{1}(j)| + |x_i^{2}(j)|$ by a counting argument. In particular, for any $i \le \ell-1,j\in [n]$, if at least one of $x_i^{1}(j)$, $x_i^{2}(j)$ is 1, then regardless of the values of $x_{\ell}^{1}(j)$, $x_{\ell}^{2}(j)$, we have $c_{i,j}\ge 1$. This can be seen from the following table. Note that by our assumptions, $x_{\ell}^{1}(j)$, $x_{\ell}^{2}(j)$ cannot both be 1. 
\begin{table}[H]
    \centering
    \begin{tabular}{|c|c|c|c|c|}
    \hline 
         $x_i^{1}(j)$ & $x_i^{2}(j)$ & $x_{\ell}^{1}(j)$ & $x_{\ell}^{1}(j)$ & $c_{i,j}$ \\
         \hline 
         1 & 0 & 0 & 0 & 2 \\
         1 & 0 & 1 & 0 & 1 \\
         1 & 0 & 0 & 1 & 3 \\
         1 & 1 & 0 & 0 & 2 \\
         1 & 1 & 1 & 0 & 1 \\
         1 & 1 & 0 & 1 & 1 \\
         0 & 1 & 0 & 0 & 2 \\
         0 & 1 & 1 & 0 & 3 \\
         0 & 1 & 0 & 1 & 1 \\
         \hline 
    \end{tabular}
    \caption{Table of possible values of the four variables and $c_{i,j}$.}
    \label{tab:values}
\end{table}
The only case that remains is $i = l$, and we observe that since $x_{\ell}^1\cap x_{\ell}^2 = 0$, $|x_{\ell}^1 + x_{\ell}^2| = |x_{\ell}^1| + |x_{\ell}^2|$. Combining this observation with the values in table~\ref{tab:values}, we see that 
\begin{align*}
|\p_2^T(c)| 
&\ge \frac{1}{2}\sum_{i=1}^{\ell}\sum_{j=1}^n |x_i^{1}(j)| + |x_i^{2}(j)| = (\sum_{i=1}^{\ell} |x_i|)/2   \\
\frac{|\p_2^T(c)| }{n\ell + 2n(\ell-1)}
&\ge \frac{\sum_{i=1}^{\ell} |x_i|}{2n\ell + 4n(\ell-1)} \ge \frac{1}{3}\frac{\sum_{i=1}^{\ell} |x_i|}{2n\ell + m(\ell-1)}.
\end{align*}
Therefore,  $Q'$ is locally testable with soundness $1/3$ for $X$-operators.
\end{proof}
Together, Lemma~\ref{lem:newpara} and~\ref{lem:newsound} gives Theorem~\ref{thm:newdistbala}.

\section*{Acknowledgements}
We thank Ted Yoder for the discussion on gauge fixing of the check product codes.

Z.H and A.N. thank Eugene Tang and the QLDPC reading group at MIT, especially Aram Harrow and Peter Shor for many helpful discussions. Z.H. also thanks Nikolas Breuckmann for many insightful conversations. Part of this work was done while Z.H. was at IBM T.J. Watson Research Center.

A.C. thanks Aram Harrow and the QLDPC reading group at MIT for helpful discussions. A.C and G.Z. are supported by the U.S. Department of Energy, Office of Science, National Quantum Information Science Research Centers, Co-design Center for Quantum Advantage (C2QA) under
contract number DE-SC0012704.

\sloppy
\printbibliography

\section*{Appendix}
\paragraph{Omitted Proofs in Section~\ref{sec:cp}}

\CoroDistTwo*
\begin{proof}
We note that the soundness follows from Lemma~\ref{lem:fact17}, and the LDPC property follows from the fact that both of its component codes are LDPC. The rate can be calculated from the number of checks: there are $2n$ qubits and $2(1-r)n$ checks, which gives rate $r$. 

The quantum distance can be seen from the following fact. Consider $Z^{(\one_v, \one_v)}$ and $X^{(\one_v, \one_v)}$ for $v\in B$. If $Z^{(\one_v, \one_v)}$ is in the stabilizer group for all $v\in B$, then $H$ must be a matrix of rank $n$. Equivalently, $H$ can be transformed into the $n\times n$ identity matrix $I$ by row operations, which means the original code $C$ is trivial. Therefore there exists some $v$ such that $Z^{(\one_v, \one_v)}$ and $X^{(\one_v, \one_v)}$ are both logical operators.
\end{proof}

\CheckProductProperty*
\begin{proof}

To prove the first claim, we see that the row space of $H_1\otimes H_2$ is exactly $C_1^{\perp}\otimes C_2^{\perp}$. Therefore, the dual of its row space is exactly the dual tensor code $C_1\otimes \FF_2^{n_2} + \FF_2^{n_1} \otimes C_2$. The rate of the new code is also easy to compute, as the number of linearly independent checks is simply $(n_1-k_1)(n_2-k_2)$.

\medskip

To argue about the distance of $C_1\star C_2$ (item {\it 2.}) is not hard either. 
Let $0\neq x \in C_1\star C_2$, and imagine 
$x$ as an $n_{1}\times n_{2}$
matrix $M$ over 
$\FF_2$.
We will have shown that $|x|\ge \min(d_{1}d_{2})$ if we find 
any $v\in \FF_2^{n_1}$ or
$w\in \FF_2^{n_2}$ such that
$|v^{T}M|\ge d_{2}$ or 
$|Mw|\ge d_{1}$ ($|.|$ is the 
Hamming weight).

Let us now run $v$ through all checks of 
$C_{1}$ (i.e. rows of $H_{1}$),
and $w$ through all checks of 
$C_{2}$ (i.e. rows of $H_{2}$).
There are two possibilities.

\begin{enumerate}
    \item Either all of the above 
    matrix-vector products give the zero vector. In that case $x$ is not only in $C_1\star C_2$, but also in the smaller
    $C_1\otimes C_2$, and therefore
    its distance from the zero 
    vector is at least
    $d_{1}d_{2}\ge \min(d_{1},d_{2})$.
    \item One of the above 
    matrix-vector products is non-zero.
    Notice that for any $v$ that is a
    $C_{1}$-check (any $w$
    that is a a $C_{2}$-check), we have
    $v^{T}M\in C_2$
    ($Mw\in C_1$), by definition, since $M$ represents an 
    $x\in C_1\star C_2$, and now we get a $\min(d_{1}d_{2})$
    bound for that reason.
\end{enumerate}
This proves that the distance is at least $\min(d_1, d_2)$. To prove equality, take $x\in C_2$ and $y\in C_1$ with minimum distances ($d_{2}$ and $d_{1}$, respectively) from zero. Then $(\underbrace{1, 0, \cdots, 0}_{n_1})\otimes x$ and $y\otimes (\underbrace{1, 0, \cdots, 0}_{n_2})$ are both valid codewords of $C_1\star C_2 = C_1\otimes \FF_2^{n_2} + \FF_2^{n_1} \otimes C_2$. This concludes the proof of our claims. \qedhere
\end{proof}

\CPDistance*
\begin{proof}
Suppose $Q$ has $n_q$ physical qubits. We have from Lemma~\ref{lem:cpprop}
\begin{align*}
d_x(Q\star C) 
&= \min \{|w|: w\in C_X\star C\sm (C_Z\star C)^{\perp}\} \\
&= \min \{|w|: w\in C_X\otimes\FF_2^n + \FF_2^{n_q}\otimes C \sm C_Z^{\perp}\otimes C^{\perp}\}.
\end{align*}
Since $C^{\perp}\ne \FF_2^{n}$, we can find $x\in C_X, |x| = d(C_X)$, and a standard basis vector $e_i\in\FF_2^{n}\sm C^{\perp}$ such that $x\otimes e_i\in C_X\otimes\FF_2^n \sm C_Z^{\perp}\otimes C^{\perp}$. Similarly, since $C_Z^{\perp}\ne \FF_2^{n_q}$, we can find $y\in C, |y| = d(C)$, and a standard basis vector $e_j\in\FF_2^{n_q}\sm C_Z^{\perp}$ such that $e_j\otimes y\in \FF_2^{n_q}\otimes C \sm C_Z^{\perp}\otimes C^{\perp}$. This shows that $d_x(Q\star C) \le \min(d(C_X), d(C))$. On the other hand, we note that from Fact~3 of Lemma~\ref{lem:cpprop},
\begin{align*}
    d_x(Q\star C) 
    &= \min \{|w|: w\in C_X\star C\sm (C_Z\star C)^{\perp}\} \\
    &\ge \min \{|w|: w\in C_X\star C\} = \min(d(C_X), d(C)).
\end{align*}
Thus $d_x(Q\star C) = \min(d(C_X), d(C))$. Similarly, we have $d_z(Q\star C) = \min(d(C_Z), d(C))$. This lemma then follows.
\end{proof}

\paragraph{Omitted Proofs in Section~\ref{sec:cpofLTCs}}
\SimpleForm*
\begin{proof}
Without loss of generality we assume
$\Pi = I_{n}$, since $\Pi$ simply permute the bits of the code without changing any of its parameters.
Given $x\in \FF_2^n$, it suffices for us to show that $\frac{n}{m}
d(x, C)/n \le |H'x|/m$. Denote the rows of $H'$ as $r_1, \cdots, r_m$, then $|H'x| = \sum_{i=1}^m r_i\cdot x$. Now suppose $r_1\cdot x = 1$, then let $x' = x + e_1$, where $e_1$ is the standard basis vector, and we see that $|H'x'| = |H'x| - 1$. We may now repeat this procedure for all rows that give syndrome-bit 1 with respect to $x$, until in 
$|H'x|$ steps we arrive at a code word. Finally we 
notice that $n/m = 1/r$.
More explicitly, let $S\subset [m]$ be the indices of rows 
violated by $x$. Define
\[
y = x + \sum_{i\in S}e_i.
\]
Then $|H'y| = 0$, since $H'x = H'(\sum_{i\in S}e_i)$. Thus, 
$d(x, C) \le d(x, y) = |S| = |H'x|$, which proves our claim.
\end{proof}

\CPTestability*
\begin{proof}
The check matrix of $C_X\star C$
is $H_X\otimes H$, and since the latter
has row and column weight bounded by $ww_X$, the bound on the 
locality
of $C_X\star C$ follows. The more interesting part is to show that $C_X\star C$ has soundness parameter at least $\rho$. 

By definition, every check 
of $C_X\star C$ is a tensor product of a check for
$C_X$ and a check for $C$.
The $i^{\rm th}$ check of $C_X$
contains the $i^{\rm th}$ bit of the checked word and some other bits
that have index beyond $m_{X}$, due to the 
$[I_{m_X}\mid h_X]$ structure of $H_{X}$. 
Given a word $x\in \FF_2^{n_x\times n}$ to be checked, we write it as $x = (x_1, \cdots, x_{n_{X}})$ where each $x_i\in \FF_2^{n}$. 
Fix $1\le i\le m_{X}$, and consider the collection 
$\Gamma_{i}$ of all those checks of $C_X\star C$ 
that have the $i^{\rm th}$ 
check of $C_X$ as their first component.
Then $|\Gamma_{i}|$ is exactly the number $t$ of all checks 
for $C$ (hence $t$ is independent of $i$).
Define
\[
y_{i} = \sum_{j\in \; i^{\rm th} \; {\rm check}\;
{\rm of} \; C_X} x_{j} \;\;\;\;\;\;\;\;\;\;\;
\left(\in\, \FF_2^{n}\right)
\]
(in particular, $j=i$ is one of the participant indices on the r.h.s.). 
Let $\Gamma_{i}'\subseteq \Gamma_{i}$ 
be the checks in $\Gamma_{i}$ that 
$x$ violates. 
These violated checks, 
due to the tensor product construction, correspond to
those checks of $C$ that $y_{i}$ violates.
Therefore, due to the soundness parameter
$\rho$ of $C$
there exists a $y_{i}'\in C$
such that $\rho |y_{i}-y_{i}'| \le \frac{n}{t}|\Gamma_{i}'|$.
Adding now $y_{i}'-y_{i}$ to $x_{i}$, while keeping all 
$x_{j}$s for $j\neq i$ the same, we have achieved that 
all checks in $\Gamma_{i}$ go through, and also we 
did not affect those checks that are not in 
$\Gamma_{i}$, since their $H_{X}$ component does 
not contain the $i^{\rm th}$ bit of $C_{X}$.
After having the above done for all $1\le i\le m_{X}$
we get a word $x'$ that passes all tests, therefore 
belongs to $C_X\star C$. Further,
\begin{equation}\label{eq:rh}
\rho |x'-x| \; \le \; 
\sum_{i=1}^{m_{X}} \, \frac{n}{t}\; |\Gamma_{i}'|
\; = \; \frac{n}{t}\;\sum_{i=1}^{m_{X}} \, |\Gamma_{i}'|
\; = \; \frac{n}{t} s
\end{equation}
where $s$ is the number of checks violated by $x$. 
Relating $s$ to $m_{X} t$, the number of all 
checks, we get from (\ref{eq:rh}) that
\[
\underbrace{\frac{s}{m_{X} t}}_{\rm probability \; of \; failed \; check}
\; \ge \;
\frac{|x'-x|}{m_{X}\cdot n}
\cdot \rho
\; = \; \underbrace{\frac{|x'-x|}{n_{X}\cdot n}}_{
\begin{array}{l}
\scaleto{\rm relative \; distance \; of}{6pt}\\[-4pt]
\scaleto{\rm of \; {\it x} \; from \; a \; code \; word}{6pt}
\end{array}
}
\cdot \underbrace{\frac{\rho \, n_X}{m_X}}_{\rm soundness}
\]
as needed.
\end{proof}

\CPQLTC*
\begin{proof}
Local testability with soundness $\rho\cdot\min(\frac{n_q}{m_X}, \frac{n_q}{m_Z})$ is proved by Lemma~\ref{lem:cpCLTC},
by applying it to both $C_{X}\star C$
and $C_{Z}\star C$ (where the check sets associated with $C_{X}$ and $C_{Z}$ are $\ol{H}_X$ and $\ol{H}_Z)$) and taking the minimum. That the minimum of the respective soundness of the two classical parts is a lower bound on the soundness of the 
CSS code is proven in Lemma \ref{lem:fact17} (originally Fact 17 in \cite{eldar2015local}).
The locality bounds also follow from Lemma~\ref{lem:cpCLTC}, where
we use the trivial $n_q$ upper bound on the row and column weights of 
$\ol{H}_X$ and $\ol{H}_Z)$.
Since $Q$ and $\ol{Q}$ are the same code, 
they have the same rate and distance, and therefore the distance of $\ol{Q}\star C$ is easily given by Lemma~\ref{lem:cpprop}. The dimension can be calculated as follows: the number of stabilizer generators in $\ol{Q}$ is $n_q - k_q$, and therefore the number of stabilizer generators in $\ol{Q}\star C$ is $(n-k)(n_q-k_q)$. 
\end{proof}





\RandomQuantumCodes*
\begin{proof}
First recall that in our construction $C\le \mathbb{F}_{2}^{n}$ is a random code with dimension $\frac{3n}{4}$ and $D\le C$ is a random subspace of $C$ of dimension
$\frac{n}{4}$. Due to the above parameters
the quantum CSS code made from classical codes $C$ and $D^{\perp}$ 
encodes $\frac{3}{4}n - \frac{1}{2}n$
qubits, therefore it has rate $1/2$.

\medskip

Next we will show
that both $C$ and $D^{\perp}$
have linear distances with high probability.
Let $\delta$ be such that $H(\delta) = 0.25$, where
$H(\delta) = \delta \log_{2}\frac{1}{\delta} + (1-\delta) \log_{2}\frac{1}{1-\delta}$ is 
the entropy function. We show that with probability very close to one
both $C$ and $D^{\perp}$ have minimum distance at least $\delta' n$, where $\delta'\rightarrow \delta$ (from below) as $n$ 
tends to infinity.

That $C$ has the above distance is simply the Gilbert-Varshamov (GV) bound  for linear codes. This says that
if we want to achieve relative distance $\delta'$, then a random linear code with rate $1-H(\delta') - \epsilon$ will achieve that, where $\epsilon$
is arbitrary small as $n$ grows. 

The proof of the GV bound is simple, and the code construction that leads to it,
telling the proof with our parameters for simplicity,
is the following: We select a random 
parity check matrix, $M$ with $\frac{n}{4}$ rows and $n$ colums. The minimum distance of $C$ is 
the minimum number of columns of $M$ that sum to zero. To argue for relative distance $\delta'$
slightly below $\delta$,
we use the well-known fact, that
\[
\log_{2} \binom{n}{\delta' n}= 
 H(\delta)\,
 (n -\omega(\log n))\;\;\;\;\;
{\rm when} \;\; \delta' \rightarrow \delta \;\; {\rm appropriately, \; as} \;\; n\rightarrow\infty
\]
The number of ways 
to select $1\le k \le \delta' n$ columns is then $\le \delta' n\cdot 2^{(H(\delta)\,
(n-\omega(\log n))}
= o
\left(2^{H(\delta)n}
\right)$, for large $n$. The probability that any select $k$
columns (no restriction on $k$) of a random matrix with $\frac{1}{4} n$ rows add up to zero is $2^{- \frac{1}{4} n}$, and we can simply finish the proof by the union bound, since
$o
\left(2^{H(\delta)n}
\right) \cdot 2^{- \frac{1}{4} n} =
o(1)$.

Let us now estimate the distance of $D^{\perp}$ in a similar way. The dimension of the linear space $D^{\perp}$ is $n-\dim D = \frac{3n}{4}$, but 
at this time we have the restriction that $D$ must be a subspace of $C$. The way to create a random subspace of $C$
of dimension $\frac{n}{4}$ is to add $\frac{n}{2}$ extra rows to the parity check matrix of $C$.
Therefore, to continue the construction, we select an additional  $\frac{1}{2}n$ rows. The added new rows constitute a matrix $N$ of dimensions $\frac{1}{2} n\times n$. When we concatenate $N$ and
$M$ (put them together into an $n\times \frac{3}{4}n$ matrix,
we get a matrix
$W = \left(\begin{array}{l} N \\\hline M\end{array}\right)$. Here $N$ and $M$ are random, and the rows of $W$ generate $D^{\perp}$. 
We calculate the probability that the rows of $W$ generate any vector of weight less than $\delta' n$, where $\delta'$
is appropriately, but only slightly less, than $\delta$.

The number of vectors with weight $\delta' n$ is upper bounded by $o
\left(2^{H(\delta)n}
\right)$ similarly to our previous 
calculation. The probability that any fixed 
select rows of $W$ (over the randomness of $W$) add up to a any particular vector $w$ is $2^{-n}$. The number of ways we can 
select a subset of rows of $w$ is $2^{ 0.75 n}$. Thus:
\begin{eqnarray*}
{\rm Prob}_{W}(\mbox{there are rows of $W$ that add up to some $w$ with weight $\le \delta n$}) \; \\
\le  \;
o
\left(2^{H(\delta)n}
\right) \cdot
2^{ 0.75 n} \cdot 2^{-n}
= o(1) \hspace{1.5in}
\end{eqnarray*}
 Thus, the probability that the minimum distance of 
$D^{\perp}$ is at least $\delta n$ is arbitrarily close to 1 when $n$ is sufficiently large.

By the union bound (at this time applied only on two terms), the probability that the minimum relative distance of either of $C$ or $D^{\perp}$ fails to be at least $\delta'$,
for a randomly chosen $W= \left(\begin{array}{l} N \\\hline M\end{array}\right)$ becomes arbitrarily small, as $n$ goes to infinity.
\end{proof}

\end{document}